\documentclass[12pt]{article}

\usepackage{amssymb,amsmath,amsfonts,amsthm}
\usepackage{graphicx}
\usepackage{bbm}

\usepackage{enumerate}

\newtheorem{proposition}{Proposition}
\newtheorem{lemma}{Lemma}

\newtheorem{corollary}[proposition]{Corollary}
\newtheorem{conjecture}{Conjecture}

\renewcommand{\c}[1]{\mathcal{#1}}
\newcommand{\g}[1]{\mathfrak{#1}}
\renewcommand{\r}[1]{\mathrm{#1}}
\newcommand{\s}[1]{\mathsf{#1}}

\newcommand{\Cx}{\mathbbm{C}}

\newcommand{\Rl}{\mathbbm{R}}

\newcommand{\idty}{\mathbbm{1}}

\DeclareMathOperator{\ran}{ran}

\DeclareMathOperator{\spa}{span}
\DeclareMathOperator*{\tr}{Tr}

\newcommand{\<}{\langle}
\renewcommand{\>}{\rangle}
\providecommand{\abs}[1]{|#1|}
\providecommand{\norm}[1]{\Vert #1 \Vert}

\newcommand{\proj}[1]{\ensuremath{|#1\rangle \langle #1|}}
\newcommand{\beq}{\begin{equation}}
\newcommand{\eeq}{\end{equation}}
\newcommand{\ket}[1]{\ensuremath{\left|{#1}\right\rangle}}
\newcommand{\bra}[1]{\ensuremath{\left\langle{#1}\right |}}
\renewcommand{\rho}{\varrho}

\begin{document}

\begin{center}
{\LARGE Matrices of fidelities for ensembles of quantum states and the Holevo quantity} \\[12pt]
Mark~Fannes$^1$, Fernando~de~Melo$^1$, Wojciech~Roga$^2$, and Karol~{\.Z}yczkowski$^{2,3}$
\end{center}

\medskip
\noindent
$^1$ Instituut voor Theoretische Fysica,
Universiteit Leuven, B-3001 Leuven, Belgium \\
$^2$ Instytut Fizyki im.~Smoluchowskiego,
Uniwersytet Jagiello{\'n}ski,
PL-30-059 Krak{\'o}w, Poland \\
$^3$Centrum Fizyki Teoretycznej, Polska Akademia Nauk,
PL-02-668 Warszawa, Poland

\medskip
\noindent
Email: \texttt{<mark.fannes@fys.kuleuven.be>}, \texttt{<fmelo@itf.fys.kuleuven.be>},\\
\texttt{<wojciech.roga@uj.edu.pl>}, and \texttt{<karol@tatry.if.uj.edu.pl>} 

\medskip\noindent
April 12, 2011

\bigskip
\noindent
\textbf{Abstract:}
The entropy of the Gram matrix of a joint purification of an ensemble of
$K$ mixed states yields an upper bound for the Holevo information
$\chi$ of the ensemble.
In this work we combine geometrical and probabilistic aspects of the ensemble in order to obtain useful bounds for $\chi$. 
This is done by constructing various correlation matrices involving fidelities between every pair of states from the ensemble.
For $K=3$ quantum states we design  a matrix of root fidelities that is positive and the entropy of which is conjectured to upper bound $\chi$.
Slightly weaker bounds are established for arbitrary ensembles.
Finally, we investigate correlation matrices involving multi-state
fidelities in relation to the Holevo quantity.


\section{Introduction} 
\label{intro} 

Quantum mechanics is a probabilistic theory described by vectors living in a complex inner product space. Due to that, ensembles of states and the linear dependencies among them are fundamental concepts of the theory. Consider an ensemble $\c E_K = \{(p_i,\ket{\varphi_i})\}$ formed by $K$ pure states, with $\ket{\varphi_i} \in \Cx^d$, $0\le p_i \le 1$, and $\sum_i^K p_i = 1$. While the statistical features of $\c E_K$ are fully described by the state $\rho = \sum_{i=1}^K p_i \proj{\varphi_i}$, all the relationships among the states are encoded in the \emph{correlation matrix} $\s C(\c E_K) := \bigl[ \sqrt{p_i p_j} \bra{\varphi_i} \varphi_j\> \bigr]_{ij}$. This correlation matrix is nothing but the Gram matrix of the sub-normalized vectors $\{\sqrt{p_i} \ket{\varphi_i}\}$. As such, $\s C(\c E_K)$ is always positive semi-definite and completely defines the space spanned by $\c E_K$.

The connection between ensembles and correlation matrices can promptly be extended to mixed states via their purifications. Let $\rho \in \c M^{d_1}$ be a state acting on a Hilbert space $\c H_1$ and let $d_2$ be the rank of $\rho$. Then, there exists a vector $\ket{\varphi}\in \Cx^{d_1\times d_2}$ acting, by virtue of the identification $\ket{\varphi} \mapsto \proj{\varphi}$, on $\c H_1 \otimes \c H_2$ such that $\rho = \tr_2 \proj{\varphi}$. The correlation matrix with elements $\sqrt{p_i p_j} \bra{\varphi_i} \varphi_j\>$ can then be constructed. Purifications are, however, not unique. If $\ket{\varphi}$  is a valid purification of $\rho$, so is $\idty \otimes U \ket{\varphi}$, with $U$ a unitary matrix. In this way, the description of an ensemble must be augmented with the choice of purifications, thus $\c E_K = \{(p_i, \rho_i, U_i)\}$.  To make explicit the dependence of $\s C(\c E_K)$ on the unitary matrices $U_i$, one simply notes that all the purifications of $\rho$ can be written as $\ket{\varphi} =  \bigl( \rho^{1/2} \otimes U^{\s T} \bigr) \sum_i \ket{ii}$, with $\s T$ the transposition taken on the basis $\{\ket{j}\}$. Using this, it is easy to obtain the general form of a correlation matrix
\beq
\s C(\c E_K) := \bigl[ \tr\rho_j^{1/2} U_j^* U_i \rho_i^{1/2} \bigr]_{ij}.
\label{cor}
\eeq

It comes as no surprise that the geometrical aspects of ensembles, as depicted by their correlation matrix, play an important role in the encoding of information into quantum systems~\cite{JS00,MJ04, RFZ10}. Consider for instance a typical scenario of quantum information communication: a sender $A$ holds an ensemble of $K$ quantum states $\c E_K =\{(p_j, \rho_i, U_i)\}$. Given an event $i$, with probability $p_i$, $A$ sends the state $\rho_i$ to a receiver $B$. The latter then performs a measurement on the delivered state aiming to distinguish it from the other states in the ensemble, and thus  to recover the index $i$. For quantum states are not necessarily orthogonal, $B$  won't generically fully succeed in this task. Holevo~\cite{holevo} showed that the information accessible to $B$ given the ensemble $\c E_K$ is upper bounded by
\beq
\chi(\c E_K) = \s S(\rho) - \sum_{i=1}^K p_i \s S(\rho_i);
\label{chi}
\eeq
with $\s S: \c M^d \to \Rl^+ : \sigma \mapsto - \tr(\sigma \log \sigma)$ the von Neumann entropy, and $\rho = \sum_{j=1}^K p_j \rho_j$. Accordingly, $\chi$ is  called the \emph{Holevo quantity}. Since the linear dependence between the states plays an important role for the information retrieval by $B$, it is expected that the geometry induced by $\c E_K$ should feature in the estimation of the accessible information. The Holevo quantity does not explicitly acknowledge that. Recently this gap has been closed and $\chi(\c E_K)$  was shown to be upper bounded by the entropy of the correlation matrix $\s C(\c E_K)$, see~\cite{RFZ10}. In particular, a tighter bound is obtained asking for the purifications that minimize the entropy of the correlation matrix, that is:
\beq
\chi(\c E_K) \le \min_{U_1,\ldots ,U_K} \s S\bigl(\bigl[ \tr \rho_j^{1/2} U_j^* U_i \rho_i^{1/2}\bigr]_{ij} \bigr).
\label{minSC}
\eeq

Nevertheless, the correlation matrix $\s C(\c E_K)$ works at the level of quantum amplitudes $\bra{\varphi_i}\varphi_j\>$. It is thus highly desirable to join the geometrical description given by the correlation matrix with the intrinsic probabilistic aspects of quantum theory. 

In fact, consider the case of $K \!= \!2$ states. The minimization in Eq.~\eqref{minSC} for an ensemble $\c E_2$ of two states can be exactly performed and the minimal entropy is obtained for the correlation matrix
\beq
\begin{bmatrix}
p_1& \sqrt{p_1 p_2} \sqrt{\s F(\rho_1;\rho_2)} \\
\sqrt{p_1 p_2} \sqrt{\s F(\rho_1;\rho_2)}& p_2
\end{bmatrix}
\eeq
where $\s F(\rho_1;\rho_2) := \bigl( \tr |\sqrt{\rho_1} \sqrt{\rho_2} | \bigr)^2$ is the \emph{fidelity} between the two states~\cite{Uh76,Uh92,Jo94}. We therefore obtain the upper bound
\beq
\chi(\c E_2) \le \s S\left( \begin{bmatrix}
p_1& \sqrt{p_1 p_2} \sqrt{\s F(\rho_1;\rho_2)} \\
\sqrt{p_1 p_2} \sqrt{\s F(\rho_1;\rho_2)}& p_2
\end{bmatrix} \right).
\label{2bound}
\eeq
The fidelity can be expressed as the maximum \emph{transition probability} between joint purifications of the two states, i.e., 
\beq
\s F(\rho_1;\rho_2) = \max |\!\bra{\varphi_1}\varphi_2\>|^2 \quad\text{with}\quad \rho_i = {\tr}_2 \proj{\varphi_i}.
\eeq 
Here, the connection between geometrical and probabilistic aspects is immediate. This is, however, only possible due to the small size of the ensemble.

In this paper, we investigate different correlation matrices based on fidelities which give upper bounds to the Holevo quantity. Surprisingly, we show that for the case $K=3$ the correlation matrix based on the square root fidelities is still positive, and numerically support that it bounds the Holevo quantity (Section~\ref{Ke3}). For  ensembles of an arbitrary number of states we construct distinct correlation matrices leading to slightly weaker bounds to $\chi$ (Section~\ref{Karb}). Among these constructions, we introduce a correlation matrix that takes into account the relationships among several states (Section~\ref{multi}). We start, however, by fixing notation and recalling some basic properties of the fidelity measure.

\section{Fidelity}
\label{fidel} 

Uhlmann extended the transition probability between pure states to fidelity between general mixed states using joint purifications, i.e., by introducing an ancillary system such that the mixed states are marginals of pure states on the composite system, see~\cite{Uh76}. As purification is not unique one has to maximize the fidelity of these pure states over all joint purifications. This leads to
\begin{align}
\s F\bigl( \rho_1; \rho_2 \bigr) 
&= \max_{\{U \mid \text{ unitary}\}}\ \bigl(\bigl| \tr \sqrt{\rho_1} \sqrt{\rho_2} U \bigr|\bigr)^2 \\
&= \bigl( \tr \bigl| \sqrt{\rho_1} \sqrt{\rho_2}\, \bigr|\bigr)^2 = \big(\tr \sqrt{ \sqrt{\rho_1} \rho_2 \sqrt{\rho_1}}\bigr)^2 = \bigl( \tr \sqrt{ \rho_1 \rho_2} \bigr)^2, 
\label{3}
\end{align}
see Lemma~\ref{lem1} for the last equality. Fidelity can be extended to general quantum probability spaces~\cite{Uh76}. Note also that in some papers the root fidelity, $\sqrt{\s F}$, is referred to as `fidelity'.

Fidelity enjoys a number of basic properties, see e.g.~\cite{Jo94}. It takes values in $[0,1]$, $ \s F\bigl( \rho_1; \rho_2 \bigr) = 1$ if and only if $\rho_1 = \rho_2$, and $\s F\bigl( \rho_1; \rho_2 \bigr) = \s F\bigl( \rho_2; \rho_1 \bigr)$. The notion of closeness based on fidelity is the same as that of the usual Kolmogorov distance
\begin{equation*}
1 - \sqrt{\s F \bigl( \rho_1; \rho_2 \bigr)} \le \norm{\rho_1 - \rho_2}_1 := \tfrac{1}{2}\, \tr \abs{\rho_1 - \rho_2} \le \sqrt{1 - \sqrt{ \s F \bigl( \rho_1; \rho_2 \bigr)}}.
\end{equation*}
Fidelity is monotonic under a quantum operation $\Gamma$ (a completely positive trace-preserving map): 
\begin{equation*}
\s F \bigl( \Gamma(\rho_1); \Gamma(\rho_2) \bigr) \ge \s F \bigl( \rho_1; \rho_2 \bigr).
\end{equation*}
Furthermore, fidelity is concave in each of its arguments
\begin{equation*}
\s F\bigl( \tfrac{1}{2}\, \rho_1 + \tfrac{1}{2}\, \rho_2; \rho_3 \bigr) \ge \tfrac{1}{2}\, \s F \bigl( \rho_1; \rho_3 \bigr) + \tfrac{1}{2}\, \s F\bigl( \rho_2; \rho_3 \bigr)
\end{equation*}
and the square root of $\s F$ is also jointly concave,
\begin{equation*}
\sqrt{ \s F\bigl( \tfrac{1}{2}\, \rho_1 + \tfrac{1}{2}\, \rho_2; \tfrac{1}{2}\, \rho_3 + \tfrac{1}{2}\, \rho_4 \bigr)} \ge \tfrac{1}{2}\, \sqrt{ \s F\bigl( \rho_1; \rho_3 \bigr)} + \tfrac{1}{2}\, \sqrt{ \s F\bigl( \rho_2; \rho_4 \bigr)}.
\end{equation*}

We end this section by proving the last expression for the fidelity in~(\ref{3}).
   
\begin{lemma}
\label{lem1}
Let $A$ and $B$ be positive semi-definite matrices then $AB$ is diagonalizable and the eigenvalues of $AB$ belong to $\Rl^+$. This implies that $\sqrt{AB}$ is uniquely defined.
\end{lemma}   

\begin{proof}
For any pair of matrices the eigenvalues of $XY$ coincide with those of $YX$ up to multiplicities of zero. Writing $AB = \sqrt A \bigl( \sqrt A B \bigr)$ and observing that $\sqrt A B \sqrt A$ is positive semi-definite we see that the eigenvalues of $AB$ belong to $\Rl^+$. 

Next we show that for any eigenvalue $\lambda$ of $AB$ 
\begin{equation*}
\ker\bigl( (AB - \lambda \idty)^2 \bigr) = \ker(AB - \lambda \idty).
\end{equation*}
Therefore, $AB$ has only trivial Jordan blocks and is diagonalizable. Suppose that $(AB - \lambda \idty)^2 \varphi = 0$. Applying $\sqrt B$ to that expression we obtain $(\sqrt B A \sqrt B - \lambda \idty)^2 \sqrt B \varphi = 0$. Therefore $(\sqrt B A \sqrt B - \lambda \idty) \sqrt B \varphi = 0$. Applying once more $\sqrt B$ to this expression we have $BAB \varphi = \lambda B \varphi$. We now substitute this in $(AB - \lambda \idty)^2 \varphi = 0$ to conclude that $\lambda AB \varphi = \lambda^2 \varphi$. Clearly, if $\lambda > 0$ we are done. So the case $\lambda = 0$ remains. Now, from $BAB\varphi = 0$ we obtain $\sqrt A B \varphi = 0$ and so $AB \varphi = 0$.
\end{proof}

We recall the general characterization of positivity for 2 by 2 block matrices 
\beq
\begin{bmatrix} x &z \\z^* &y \end{bmatrix} \ge 0 \text{ iff } x,y \ge 0, \text{ and } z = \sqrt x\, u \sqrt y \text{ with } \norm u \le 1.
\label{+def}
\eeq 

Let $b>0$, $a,c \ge 0$, and $\norm u \le 1$ and put
\begin{equation*}
A = \begin{bmatrix} a &\sqrt a u \sqrt c \\ \sqrt c u^* \sqrt a &c \end{bmatrix} \quad\text{and}\quad B = \begin{bmatrix} b &0 \\ 0 &0 \end{bmatrix},
\end{equation*}
then
\begin{equation*}
\sqrt{AB} = \begin{bmatrix} x &0 \\y &0 \end{bmatrix}
\end{equation*}
where
\begin{equation*}
\begin{split}
&x = \sqrt{ab} = \frac{1}{\sqrt b}\, \sqrt{\sqrt b a \sqrt b}\, \sqrt b \quad\text{and} \\
&y = 0 \text{ on } \ker(a) \quad\text{and}\quad y \sqrt{ab} = \sqrt c u^* \sqrt a b.
\end{split}
\end{equation*}
The last equation defines $y$ uniquely on $\ran(\sqrt{ab}) = \ran(a)$ because $\sqrt{ab} \varphi = 0$ implies $\sqrt ab\varphi = 0$.

\section{Ensembles of three states}
\label{Ke3} 

Following the case of 2 states, explained in Section~\ref{intro}, one might expect that for an ensemble of 3 states $\c E_3 = \{(p_1,\rho_1,U_1), (p_2,\rho_2,U_2), (p_3,\rho_3,U_3)\}$ the minimum of $\s S\bigl( \s C(\c E_3) \bigr)$ would be reached for a matrix of the form
\beq
\s C_{\sqrt{\s F}}(\c E_3) = \begin{bmatrix}
p_1 &\sqrt{p_1 p_2} \sqrt{\s F_{12}} &\sqrt{p_1 p_3} \sqrt{\s F_{13}} \\
\sqrt{p_1 p_2} \sqrt{\s F_{12}} &p_2 &\sqrt{p_2 p_3} \sqrt{\s F_{23}} \\
\sqrt{p_1 p_3} \sqrt{\s F_{13}} &\sqrt{p_2 p_3} \sqrt{\s F_{23}} &p_3
\end{bmatrix}.
\label{fid3}
\eeq
Here we introduced the simplified notation $\s F_{ij} = \s F\bigl(\rho_i; \rho_j\bigr)$. Such a matrix would have the largest possible off-diagonal elements and would therefore be rather pure and hence have a low entropy. Unfortunately there are constraints between the unitaries appearing in a correlation matrix and these are in general simply too strong to turn~(\ref{fid3}) 
into a correlation matrix. Indeed, numerical minimization over the unitary matrices suggests that for some instances $\s S\bigl( \s C_{\sqrt{\s F}}(\c E_3) \bigr) < \min_{U_1,U_2,U_3} \s S(\s C(\c E_3))$, see Figure~\ref{fig:violation}. Therefore not every root fidelity matrix (\ref{fid3}) represents a correlation matrix (\ref{cor}).

Note that for an ensemble $\c E_K = \{(p_i,\rho_i)\}$ of $K$ states the matrix 
\begin{equation}
\s E_K(\{\rho_i\}) = \bigl[ \sqrt{\s F_{ij}} \bigr]_{ij}
\label{fidmat}
\end{equation}
depending only on the $K$ states $\{\rho_i\}$ has the same number of positive, negative, and zero eigenvalues as the correlation matrix $\s C_{\sqrt{\s{F}}}(\c E_K)$. In fact, both matrices are connected via the $*$-congruence $\s D(\c E_K) = \bigl[ \delta_{ij}/\sqrt{p_i} \bigr]_{ij}$, i.e., $\s E_K(\{\rho_i\}) = \s D(\c E_K)\, \s C_{\sqrt{\s F}}(\c E_K)\, \s D(\c E_K)^*$ and inertia applies. For $K=3$, even the positivity of the matrix $\s C_{\sqrt{\s F}}(\c E_3)$ or equivalently of $\s E_3(\{\rho_j\})$ is not obvious at first sight. Remarkably, positivity still holds for $K=3$ states, but fails in general for larger sets. This is the content of the following proposition, obtained in collaboration with D.~Vanpeteghem~\cite{FV04}.

\begin{figure}[h!tf]
\begin{center}
\scalebox{.85}{\includegraphics{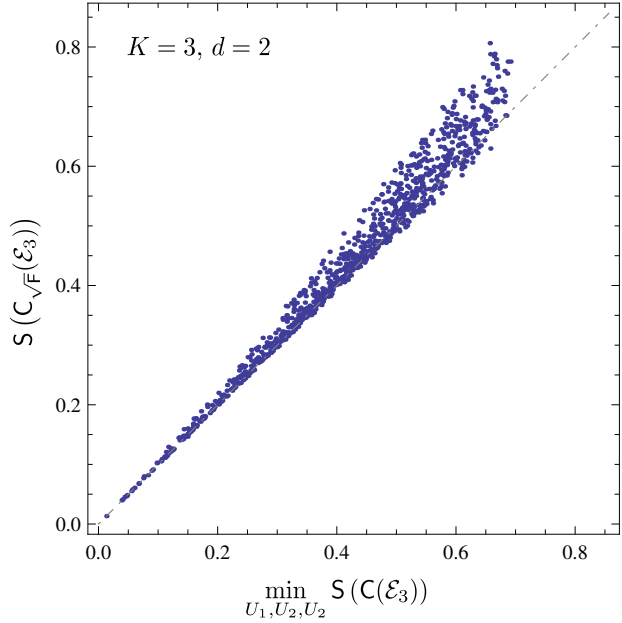}}
\end{center}
\caption{Entropy of the matrix (\ref{fid3}) of square root fidelities can be smaller than the minimal entropy of a correlation matrix (\ref{minSC})}
\label{fig:violation}
\end{figure}

\begin{proposition}
\label{thm1}
For any set $\{\rho_1, \rho_2, \rho_3\}$ with $K \!= \!3$ states the matrix
\beq
\s E_3(\{\rho_j\}) = \bigl[ \sqrt{\s F_{ij}} \bigr]_{ij},
\label{sqrf}
\eeq
is positive semi-definite. For $K \ge 4$,  $\s E_K(\{\rho_j\})$ is not positive in general.
\end{proposition}

\begin{proof}
As the matrix~(\ref{sqrf}) is real symmetric, and each of its two-dimensional principal minors is non-negative, for $K=3$ it suffices to prove that its determinant is non-negative. To do that we must show that
\beq
 \s F_{12} +  \s F_{13} +  \s F_{23} \le 1 + 2 \sqrt{\s F_{12} \s F_{13} \s F_{23}}.
\label{thm1:2}
\eeq
We may assume that the $\rho_j$ are faithful (invertible). The general case is obtained by continuity. Denoting $\sqrt{\s F_{12}}$ by $a$, we certainly have that $0 \le a \le 1$. Therefore
\begin{equation*}
0 \le \s F_{13} + \s F_{23} - 2 a \sqrt{\s F_{13} \s F_{23}}
\end{equation*}
and we wish to upper bound this by $1 - a^2$. Using the polar decomposition and the assumption on the supports of the $\rho_j$, there exist unitary matrices $U_1$ and $U_2$ such that
\begin{equation*}
\sqrt{\rho_1} \sqrt{\rho_3} = U_1\, \bigl| \sqrt{\rho_1} \sqrt{\rho_3} \bigr|
\enskip\text{and}\enskip
\sqrt{\rho_2} \sqrt{\rho_3} = U_2\, \bigl| \sqrt{\rho_2} \sqrt{\rho_3} \bigr|.
\end{equation*}
We now express the root fidelities as follows
\begin{equation*}
\begin{split}
&\sqrt{\s F_{13}} = \tr \bigl| \sqrt{\rho_1} \sqrt{\rho_3} \bigr| = \tr U_1^* \sqrt{\rho_1} \sqrt{\rho_3} \\
&\sqrt{\s F_{23}} = \tr \bigl| \sqrt{\rho_2} \sqrt{\rho_3} \bigr| = \tr U_2^* \sqrt{\rho_2} \sqrt{\rho_3}.
\end{split}
\end{equation*}
Using the Hilbert-Schmidt scalar product the fidelities become
\begin{equation*}
\sqrt{\s F_{13}} = \bra f h \>_{\r{HS}}
\enskip\text{and}\enskip
\sqrt{ \s F_{23}} = \bra g h \>_{\r{HS}}
\end{equation*}
with
\begin{equation*}
\ket f := \sqrt{\rho_1}\, U_1,\enskip \ket g := \sqrt{\rho_2}\, U_2,\enskip \text{and } \ket h := \sqrt{\rho_3}.
\end{equation*}
We can then verify the following properties
\begin{align*}
&\bra f h \>_{\r{HS}} = \abs{\!\bra f h \>_{\r{HS}}},\enskip \bra g h \>_{\r{HS}} = \abs{\!\bra g h \>_{\r{HS}}}, \\
&\norm f_{\r{HS}} = \norm g_{\r{HS}} = \norm h_{\r{HS}} = 1,\enskip \text{and} \\
&\abs{\!\bra f g \>} = \abs{\tr U_1^* \sqrt{\rho_1} \sqrt{\rho_2} U_2} \le 
\sup_{W\ \text{unitary}}\ \abs{\tr \sqrt{\rho_1} \sqrt{\rho_2}\, W} = \sqrt{\s F_{12}} = a.
\end{align*}
The first statement of the proposition now follows from Lemma~\ref{lem2}, the proof of which can be found in the Appendix.

It is not difficult to numerically find sets with four states for which $\s E_4(\{\rho_j\}) \not\ge 0$. Using these states as a subset for larger $K$ renders the positivity of $\s E_K$ in general impossible.
\end{proof}

\begin{lemma}
\label{lem2}
Let $\ket f$ and $\ket g$ be normalized vectors in a Hilbert space $\c H$ and let $a$ be such that $\abs{\!\bra f g \>} \le a \le 1$, then
\begin{equation*}
\begin{split}
&\sup_{\ket h,\ \norm h \le 1} \Bigl( \abs{\!\bra f h \>}^2 + \abs{\!\bra g h \>}^2 - 2 a \abs{\!\bra f h \>} \abs{\!\bra g h \>} \Bigr) \\
&\qquad = (1 - a) (1 + \abs{\!\bra f g \>}) \le (1 - a^2).
\end{split}
\end{equation*}
\end{lemma}

Positivity of the matrix $\s E_K(\{\rho_i\})$ of root fidelities~(\ref{fidmat}) for $K \!= \!3$ provides us with a simple continuity estimation.  Suppose that $\s F_{23} = 1$, then equation~\eqref{thm1:2} immediately leads to $\s F_{12} = \s F_{13}$, as expected. The next corollary gives a smooth version of this observation:

\begin{corollary}
\label{cor1}
\beq
\bigl| \sqrt{\s F_{12}} - \sqrt{\s F_{13}} \bigr| \le \sqrt{1 - \s F_{23}}
\label{jeden}
\eeq
\beq
\bigl|\s F_{12} - \s F_{13} \bigr| \le 2 \sqrt{1 - \s F_{23}}.
\label{dwa}
\eeq
\end{corollary}

\begin{proof}
We can rewrite inequality~(\ref{thm1:2}) as  
\begin{equation*}
\begin{split}
\s F_{12} + \s F_{13} - 2 \sqrt{\s F_{12} \s F_{13}} 
&\le 1 + 2 \sqrt{\s F_{12} \s F_{13} \s F_{23}} - \s F_{23} - 2 \sqrt{\s F_{12} \s F_{13}} \\
&\le (1 - \s F_{23}) - 2 \sqrt{\s F_{12} \s F_{13}} (1 - \sqrt{\s F_{23}}) \\
&\le (1 - \s F_{23}).
\end{split}
\end{equation*}
Taking the square root of both sides one arrives at the first statement (\ref{jeden}) of the corollary, which implies the second one (\ref{dwa}). These statements show that the fidelity between quantum states is continuous with respect to a variation of one of its arguments. 
\end{proof}

Now that we have established the positivity of $\s E_3(\{\rho_i\})$ and hence of 
\beq
\s C_{\sqrt{\s F}}(\c E_3) := \bigl[ \sqrt{p_i p_j}\, \sqrt{\s F_{ij}}\bigr]_{ij}
\label{CrootFid}
\eeq 
we can try to link the entropy of $\s C_{\sqrt{\s F}}(\c E_3)$ to the corresponding Holevo quantity. However, Figure~\ref{fig:violation} and the fact that the positivity is only possible for very small ensembles ($K \le 3$) already point to a hard to ascertain connection. We formulate the following conjecture that is well-supported by numerical evidence:

\begin{conjecture}
\label{conj}
For any ensemble $\c E_3 = \{(p_i,\rho_i)\}$ with $K \!= \!3$ three states of arbitrary dimension $d$ the corresponding Holevo quantity $\chi(\c E_3)$ is bounded from above by the entropy of the matrix $\s C_{\sqrt{\s F}}(\c E_3)$ defined in~(\ref{CrootFid}):
\beq
\chi\bigl( \c E_3 \bigr) \le \s S\bigl( \s C_{\sqrt{\s F}}(\c E_3) \bigr).
\label{eq:conj}
\eeq 
\end{conjecture}
 
\medskip
This conjecture is easily seen to be true for the particular case of ensembles of pure states. In this case, the Holevo quantity~\eqref{chi} is equal to $\s S\bigl( \sum_i p_i \proj{\varphi_i} \bigr)$, which  in turn is equal to $\s S\bigl( \bigl[ \sqrt{p_i p_j} \bra{\varphi_i}\varphi_j\> \bigr]_{ij} \bigr)$, for both matrices have, up to multiplicities of zero, the same spectrum. The matrix of fidelities, $\s C_{\sqrt{\s F}}$, is obtained from $\bigl[ \sqrt{p_i p_j} \bra{\varphi_i}\varphi_j\> \bigr]_{ij}$ by simply taking the absolute value of all its entries. This procedure can only increase the determinant, increasing thus the entropy~\cite{MJ04}.

\begin{figure}[h!tf]
\centering
\scalebox{.65}{\includegraphics{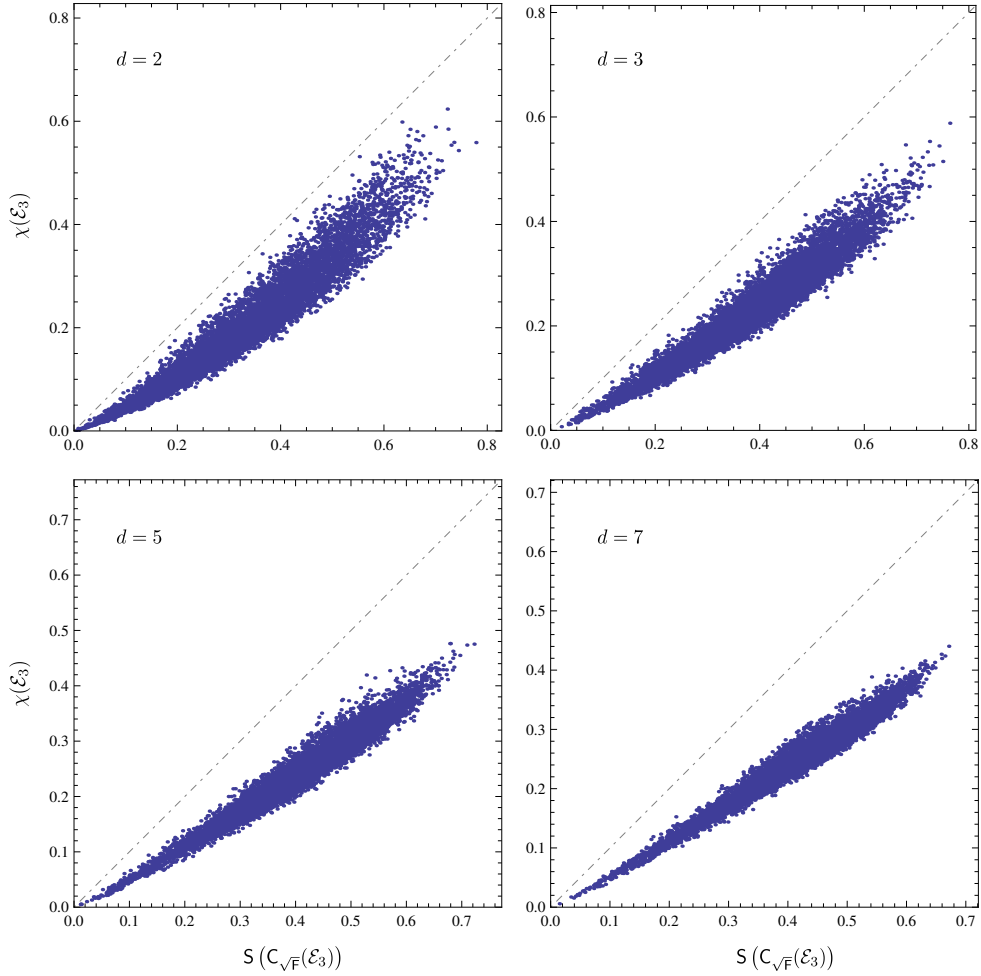}}
\caption{10000 ensembles of $K\!=\!3$ random states distributed according to the Hilbert--Schmidt measure~\cite{ZS01} and plotted on the plane $\left( \s S\bigl( \s C_{\sqrt{\s F}}(\c E_3) \bigr),\chi(\c E_3) \right)$ for $d=2,3,5,7$ support Conjecture~\ref{conj}.
}
\label{fig:conj}
\end{figure}

For the general case of mixed states extensive numerical studies support the conjecture, as partially shown in Figure~\ref{fig:conj}. Some further particular cases, and weaker forms of this conjecture can be proven. 

\begin{proposition}
\label{triples}
Let $\c E_3 = \{(p_1, \rho_1),(p_2, \rho_2),(p_3, \rho_3)\}$ be an ensemble with $K=3$ states. Then, the corresponding Holevo quantity is upper bounded by the weighted sum of the entropies of $\s C_{\sqrt{\s F}}(\c E_2)$, with $\c E_2$ all the possible two states sub-ensembles $\c E_2 = \{(p_i/(p_i+p_j),\rho_i),(p_j/(p_i+p_j),\rho_j)\}$, $i,j= 1,2,3$. That is,
\beq
\chi(\c E_3) \le \sum_{i<j} (p_i + p_j) \s S\bigl( \s C_{\sqrt{\s F}}(\c E_2)\bigr),
\label{triple}
\eeq
where $\s C_{\sqrt{\s F}}(\c E_2) = \bigl[ \sqrt{p_i p_j} \sqrt{\s F(\rho_i; \rho_j)}/(p_i+p_j) \bigr]_{ij}$.
\end{proposition}

\begin{proof}
We start by applying the bound~(\ref{2bound}) twice:
\begin{align*}
&\s S(p_1\rho_1 + p_2\rho_2 + p_3 \rho_3) - p_1 \s S(\rho_1) - p_2 \s S(\rho_2) - p_3 \s S(\rho_3) \\
&\le \left\{ \s S \Bigl( p_1 \rho_1 + (1-p_1) \frac{p_2 \rho_2 + p_3 \rho_3}{p_2 + p_3} \Bigr) - p_1  \s S(\rho_1) - (1-p_1) \s S\Bigl( \frac{p_2 \rho_2 + p_3 \rho_3}{p_2 + p_3} \Bigr) \right\} \\
&\quad+ (p_2 + p_3) \left\{ \s S\Bigl( \frac{p_2 \rho_2 + p_3 \rho_3}{p_2 + p_3} \Bigr) - \frac{p_2}{p_2 + p_3}\, \s S(\rho_2) - \frac{p_3}{p_2 + p_3}\,  \s S(\rho_3) \right\} \\[6pt]
&\le  \s S \left( \begin{bmatrix}
p_1 &\sqrt{p_1 (1-p_1)} \sqrt{\s F\bigl( \rho_1;\frac{p_2 \rho_2 + p_3 \rho_3}{p_2 + p_3}} \bigr) \\
\sqrt{p_1 (1-p_1)} \sqrt{\s F\bigl( \rho_1;\frac{p_2 \rho_2 + p_3 \rho_3}{p_2 + p_3} \bigr)} &1-p_1 \end{bmatrix} \right) \\[6pt]
&\quad+ (p_2 + p_3) \s S \left( \frac{1}{p_2 + p_3}\,
\begin{bmatrix}
p_2 &\sqrt{p_2 p_3} \sqrt{\s F(\rho_2;\rho_3) }\\
\sqrt{p_2 p_3} \sqrt{\s F(\rho_2;\rho_3)} &p_3 \end{bmatrix} \right).
\end{align*}
In order to prove~(\ref{triple}) it therefore suffices to show that
\begin{multline}
\s S \left( \begin{bmatrix}
p_1 &\sqrt{p_1 (1-p_1)} \sqrt{\s F\bigl( \rho_1;\frac{p_2 \rho_2 + p_3 \rho_3}{p_2 + p_3} \bigr)} \\
\sqrt{p_1 (1-p_1)}\sqrt{ \s F\bigl( \rho_1;\frac{p_2 \rho_2 + p_3 \rho_3}{p_2 + p_3} \bigr)} &1-p_1 \end{bmatrix} \right) \\[6pt]
\shoveright{\le (p_1 + p_2)\,  \s S \left( \frac{1}{p_1 + p_2}\,
\begin{bmatrix}
p_1 &\sqrt{p_1 p_2} \sqrt{ \s F(\rho_1;\rho_2)} \\
\sqrt{p_1 p_2} \sqrt{\s F(\rho_1;\rho_2)} &p_2 \end{bmatrix} \right) \phantom{.}} \\[6pt]
+ (p_1 + p_3)\, \s S \left( \frac{1}{p_1 + p_3}\,
\begin{bmatrix}
p_1 &\sqrt{p_1 p_3} \sqrt{\s F(\rho_1;\rho_3)} \\
\sqrt{p_1 p_3} \sqrt{\s F(\rho_1;\rho_3)} &p_3 \end{bmatrix} \right) .
\label{bound}
\end{multline}

Let $\sigma$ be a qubit density matrix with determinant $D$. As the eigenvalues of $\sigma$ are completely determined by $D$, the entropy of $\sigma$ is a function of $D$. Using the well-known integral representation
\beq
\log x = \int_0^\infty \!dt\, \Bigl( \frac{1}{t+1} - \frac{1}{t+x} \Bigr),\enskip x > 0
\eeq
of the logarithm we easily obtain
\beq
\s S(\sigma) = f(D) := \int_0^\infty \!dt\, \frac{D(2t+1)}{(t+1)(t^2+t+D)},\enskip 0 < D < \tfrac{1}{4}.
\label{f}
\eeq
It is obvious from this expression that the function $f$ can be extended to a continuous function on $\Rl^+$ that is smooth on $]0,\infty[$. The function $f$ is monotonically increasing and satisfies the following inequality
\beq
f(a^2 x + b^2 y) \le a f(x) + b f(y),\enskip 0 \le a,b \le 1 \text{ and } x,y \in \Rl^+.
\eeq
The proof of this inequality is rather tedious. It relies on the inequality $f(x) \le 2x f'(x)$ for $x \in \Rl^+$ which can be verified by explicitly performing the integration in~(\ref{f}).

We now apply this inequality with the choices
\beq
a = p_1 + p_2,\enskip b = p_1 + p_3
\eeq
and for $x$ and $y$ we choose the determinants of the density matrices appearing in the right hand side of~(\ref{bound}):
\beq
x = \frac{p_1 p_2}{(p_1 + p_2)^2}\, (1 - \s F(\rho_1;\rho_2)),\enskip y = \frac{p_1 p_3}{(p_1 + p_3)^2}\, (1 - \s F(\rho_1;\rho_3)).
\eeq
The proof then follows from the monotonicity of $f$ and the concavity of fidelity.
\end{proof}

\begin{proposition}
\label{propo1}
Let $\c E_3 = \{(p_i,\rho_i)\}$  be an ensemble consisting of three states. Then the Holevo quantity~\eqref{chi} is bounded by the entropy of the Hadamard product between the matrices $\s C_{\sqrt{\s F}}(\c E_3)$, defined in~(\ref{CrootFid}) and $\s T_b = [b + (1-b)\delta_{ij}]_{ij}$, with $0 \le b \le \frac{1}{2}$, 
\beq
\chi(\c E_3) \le \s S( \s C_{\sqrt{\s F}}(\c E_3) \circ \s{T}_b).
\label{pro1}
\eeq
\end{proposition}

\begin{proof}
The way of reasoning is similar to that given in \cite{RFZ10}. The first stage of the proof requires finding a suitable three partite state $\omega_{123}$. In the second stage, the strong sub-additivity of von Neumann entropy is employed. To obtain the Holevo quantity on one side of the inequality, the diagonal blocks of $\omega_{123}$ should read $|ii\>\<ii|\otimes p_i\rho_i$. Then, the left hand side of the strong sub-additivity relation, written in the form \cite{R02, nielsen}
\beq
\s S(\omega_3) + \s S(\omega_1) - \s S(\omega_{23}) \le \s S(\omega_{12}),
\label{SSA}
\eeq
leads to the Holevo quantity of the ensemble $\c E_3$. Here we use the notation in which, for instance, $\omega_{12}$ denotes the partial trace of $\omega_{123}$ over the third subsystem. The state $\omega_{12}$ depends on off-diagonal blocks of $\omega_{123}$ and its entropy provides an upper bound for the Holevo quantity.

According to this scheme the proof of Proposition~\ref{propo1} goes as follows.
Assume that the states $\rho_i$ are invertible and consider the state
\beq
\Omega = \bigoplus_{\{i,j\}_{i,j=1,2,3}} \sum_{\ell,m \in \{i,j\}} |\ell\ell\>\<mm|\otimes\frac{1}{2}\sqrt{p_\ell p_m} \sqrt{\rho_\ell \rho_m}.
\label{sumblock}
\eeq
Because of Lemma~\ref{lem1} we may take limits for the general case. 
This matrix is positive, since it is built by permutations of block matrices of size $2d$
\beq
X = \begin{bmatrix}
p_i \rho_i &\sqrt{p_ip_j} \sqrt{\rho_i\rho_j} \\
\sqrt{p_ip_j} \sqrt{\rho_j\rho_i} &p_j\rho_j
\end{bmatrix}.
\label{positiveblock}
\eeq

To show that $X$ is positive we use~(\ref{+def}). We write $\sqrt{\rho_i \rho_j} = \sqrt{\rho_i} u \sqrt{\rho_j}$ and so $u = \rho_i^{-1/2} \sqrt{\rho_i \rho_j} \rho_j^{-1/2}$. We still have to show that $\norm u \le 1$. This follows from the C*-property $\norm{X^*}^2 = \norm X^2 = \norm{X^*X}$ of the norm:
\beq
\begin{split}
\Bigl\Vert \frac{1}{\sqrt{\rho_i}}\, \sqrt{\rho_i \rho_j}\, \frac{1}{\sqrt{\rho_j}} \Bigr\Vert^2 
&= \Bigl\Vert \frac{1}{\sqrt{\rho_i}}\, \sqrt{\rho_i \rho_j}\, \frac{1}{\rho_j}\, \sqrt{\rho_j \rho_i}\, \frac{1}{\sqrt{\rho_i}} \Bigr\Vert \\
&= \Bigl\Vert \frac{1}{\sqrt{\rho_i}}\, \frac{1}{\sqrt{\rho_j}}\, \sqrt{\rho_j}\,\sqrt{\rho_i \rho_j}\, \frac{1}{\sqrt{\rho_j}}\,\frac{1}{\sqrt{\rho_j}}\, \sqrt{\rho_j \rho_i}\, \sqrt{\rho_j}\, \frac{1}{\sqrt{\rho_j}}\, \frac{1}{\sqrt{\rho_i}} \Bigr\Vert \\
&= \Bigl\Vert \frac{1}{\sqrt{\rho_i}}\, \frac{1}{\sqrt{\rho_j}}\, \Bigl( \sqrt{\sqrt{\rho_j}\, \rho_i \sqrt{\rho_j}} \Bigr)^2\, \frac{1}{\sqrt{\rho_j}}\, \frac{1}{\sqrt{\rho_i}} \Bigr\Vert \\
&= \Bigl\Vert \frac{1}{\sqrt{\rho_i}}\, \rho_i\, \frac{1}{\sqrt{\rho_i}} \Bigr\Vert = 1.
\end{split}
\label{normeq}
\eeq
Therefore, since $\Omega$ in (\ref{sumblock}) is positive, also is its partial trace:
\beq
\sum_{i,j=1}^3 |ii\>\<jj| \otimes \frac{1}{2}\, \sqrt{p_ip_j} \sqrt{\rho_i\rho_j} + \sum_{i=1}^3 |ii\>\<ii| \otimes \frac{1}{2}\, p_i\rho_i\;,
\eeq
which we use as the ansatz state $\omega_{123}$. The reduced state $\omega_{12}$ is then given by
\beq
\frac{1}{2}\, \sum_{i,j} (\sqrt{p_i p_j} + \delta_{ij} p_i) \tr(\sqrt{\rho_i \rho_j}) \ket{ii}\!\bra{jj}\;,
\eeq
which has the same entropy as the correlation matrix $\s C_{\sqrt{\s F}}(\c E_3) \circ T_{1/2}$. This completes the proof for the case $b=1/2$.

To prove~(\ref{pro1}) for $0 \le b < \frac{1}{2}$ one can simply multiply the off-diagonal matrices of~(\ref{sumblock}) by a positive number smaller than 1, without changing the proof.
\end{proof}

For the special case of three qubit states the range of $b$ can be pushed up to $\frac{1}{\sqrt 3}$.  Moreover, for two pure and one mixed qubit states Conjecture~\ref{conj} is found to hold for the uniform probability distribution, $p_j = \frac{1}{3}$. These proofs contain elementary but  lengthy computations and can be found in~\cite{PhD}.

\section{General ensembles}
\label{Karb} 

Clearly, the positivity of $\s E_K(\{\rho_i\})$ for $K$ states would provide a lot of information on relations between fidelities. Unfortunately, as shown in Conjecture~\ref{thm1}, this matrix already fails to be positive for four states. This somehow points at off-diagonal elements being too large. There are several possibilities to improve this situation such as scaling down the off-diagonal entries. Doing this considerably weakens the positivity result, e.g., if $\s F_{12} = 1$, then the $(1,2)$ principal sub-matrix of a rescaled $\s E_K(\{\rho_i\})$ has no longer an eigenvalue $0$, while we still have $\rho_1 = \rho_2$, and hence $\s F_{1k} = \s F_{2k} $. In particular, positivity of such a rescaled matrix of fidelities with one of the fidelities equal to 1 does not collapse to positivity of a rescaled matrix of fidelities with one row and column removed. A better way to decrease off-diagonals is therefore to consider a matrix of powers of fidelities. 

The following proposition introduces the family $\s E^\alpha_K$ of fidelity matrices and sets a lower bound on $\alpha$ if the matrix is to be positive for any ensemble. 

\begin{proposition}
Suppose that for any set $\{\rho_i\}$ consisting of $K$ states the matrix
\beq
\s E^\alpha_K(\{\rho_i\}) := \bigl[ \s F^\alpha_{ij} \bigr]_{ij} = \bigl[ \bigl( \s F(\rho_i;\rho_j) \bigr)^\alpha \bigr]_{ij}
\eeq
is positive definite. Then $\alpha \ge 1$.
\end{proposition} 

\begin{proof}
The proof uses a particular choice of ensemble. We choose $2n + 2$ pure states determined by $\{|\varphi_i\rangle \}$ in $\Cx^n$ and consider the limit of large $n$. Let $H$ be a complex unitary Hadamard matrix~\cite{TZ06} of dimension $n$ so that
\beq
\abs{H_{ij}} = \frac{1}{\sqrt n},\ i,j=1,\dots ,n
\eeq
Define now
\beq
|\varphi_{2i-1} \rangle  = e_i \enskip\text{and}\enskip | \varphi_{2i+2}\rangle  = f_i,\enskip i = 1,2,\ldots,n,
\eeq
where $\{e_i\}$ is the standard basis in $\Cx^n$ and $\{f_i\}$ is the mutually unbiased basis of columns in $H$. We then have
\beq
\begin{split}
&\norm{\varphi_i} = 1, \\
&\< \varphi_i \,| \, \varphi_j \> = 0 \enskip\text{when $i + j$ is even and larger than 0} \\
&\abs{\< \varphi_i \, | \, \varphi_j \>} = \frac{1}{\sqrt n} \enskip\text{otherwise}.
\end{split}
\eeq
Let us express the positivity of
\beq
\bra\omega \s E^\alpha_K(\{\rho_i\})\, \omega\>
\eeq
where $\omega$ is the vector with entries $\omega_i = (-1)^i$. This leads to
\beq
2n - n^2\, \frac{1}{n^\alpha} \ge 0
\eeq
which implies in the limit $n \to \infty$ that $\alpha \ge 1$.
\end{proof}

With this result in hands, we look for the connection with the Holevo quantity.

\begin{proposition}
\label{prop:fid:square}
Let $\c E_K = \{(p_i, \ket{\varphi_i})\}$ be an ensemble of $K$ pure states. Then, 
\begin{enumerate}
\item
$\s C_{\s F}(\c E_K) := \bigl[ \sqrt{p_i p_j} \abs{\!\bra{\varphi_i} \varphi_j \>}^2 \bigr]_{ij}$ is positive
\item
the corresponding Holevo quantity $\chi(\c E_K)$ is upper-bounded by the entropy of $\s C_{\s F}(\c E_K)$,
\beq
\chi(\c E_K) \le \s S\bigl( \s C_{\s F}(\c E_K) \bigr).
\eeq
\end{enumerate}
\end{proposition}

\begin{proof}
Part \emph{1.} of the proposition is easily verified by defining the Gram matrix of the vectors $\sqrt[4]{p_i}\ket{\varphi_i}$,
\begin{equation*}
\s G(\c E_K) = \bigl[ \sqrt[4]{p_i p_j}\bra{\varphi_i}\varphi_j\> \bigr]_{ij},
\end{equation*}
and noting that $\s C_{\s F}(\c E_K)$ is equal to the Hadamard product between $\s G(\c E_K)$ and $\s G(\c E_K)^*$, i.e, $\s C_{\s F}(\c E_K) = \s G(\c E_K) \circ \s G(\c E_K)^*$.

The second part of the proposition follows the general scheme described in the beginning of proof of Proposition~\ref{propo1}. The ansatz multi-partite state $\omega_{123}$ is defined by
\beq
\omega_{123} := \sum_{ij}|ii\>\<jj|\otimes\sqrt{p_ip_j} \<\varphi_i|\varphi_j\> |\varphi_i\>\<\varphi_j|.
\eeq
Positivity of $\omega_{123}$ can be assured by realizing that such state is obtained as the partial trace of the Gram matrix $\Omega$
\beq
\Omega := \sum_{ij} |ii\>\<jj|\otimes\sqrt{p_ip_j} |\varphi_i \otimes \bar{\varphi}_i\> \<\varphi_j \otimes \bar{\varphi}_j|;
\eeq
where we introduced the transformation $\ket{\varphi} \mapsto \ket{\overline{\varphi}}$, which is realized by taking the complex conjugate of all coordinates of the state in a given basis.

Taking the state $\omega_{123}$ as an initial multi-partite state, and using the strong sub-additivity relation in the form~(\ref{SSA}) leads to the desired result.
\end{proof}

The extension of this proposition to an arbitrary ensemble $\c E_K = \{(p_i,\rho_i)\}$ of $K$ mixed states is delicate. Given Proposition~\ref{thm1}, $\s C_{\s F}(\c E_3)$ for three states is positive, as $\s C_{\s F}(\c E_3) = \s C_{\sqrt{\s F}}(\c E_3)\circ \s C_{\sqrt{\s F}}(\c E_3)$. For four mixed states positivity seems also to hold, but for ensembles with $K \ge 5$ mixed states it fails. See Figure~\ref{fig:smallestEval}.

\begin{figure}[h!tf]
\begin{center}
\scalebox{.75}{\includegraphics{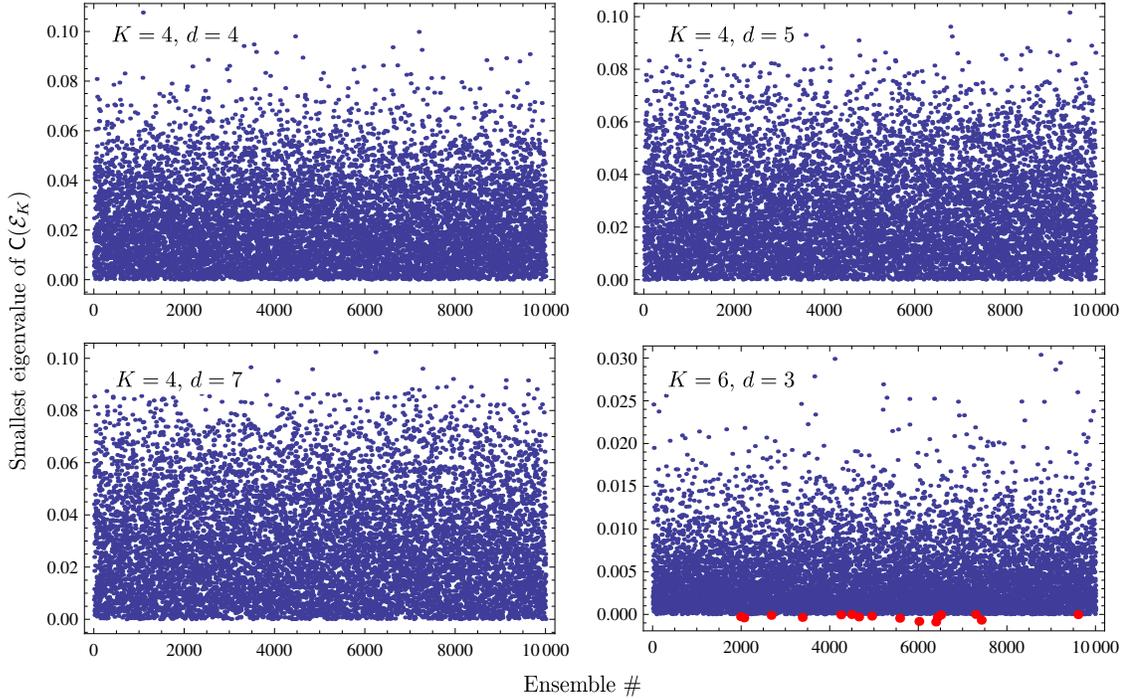}}
\end{center}
\caption{Testing the positivity of matrices of fidelities. Numerical data for ensembles of length $K$ of $d$ dimensional states 
showing that in general the fidelity matrix $\s C_{\s F}(\c E_K)$ is not positive.}
\label{fig:smallestEval}
\end{figure}

In fact, the positivity of $\s C_{\s F}$ is dimension dependent. Mixed qubit states are targeted by the next proposition.

\begin{proposition}
For any ensemble $\c E_K = \{(p_i,\rho_i)\}$ of $K$ qubit states
\begin{enumerate}
\item
the correlation matrix $\s C_{\s F}(\c E_K)$ is positive semi-definite
\item
the corresponding Holevo quantity is bounded by the entropy of $\s C_{\s F}(\c E_K)$
\beq
\chi(\c E_K) \le \s S(\s C_{\s F}(\c E_K)).
\label{pro4}
\eeq
\end{enumerate}
\end{proposition}

\begin{proof}
Construct the positive block matrix $W$  in the following way:
\begin{equation*}
W = \sum_{i,j=1}^K |ii\>\<jj| \otimes M_i M_j^*
\end{equation*}
where $M_i = \sqrt{p_i}(\rho_i,  \sqrt{\det{\rho_i}} \idty)^{\s T}$.
Then the block elements of $W$ read:
\begin{equation*}
W_{ij} = |ii\>\<jj| \otimes \sqrt{p_i p_j} (\rho_i \rho_j + \sqrt{\det{\rho_i \rho_j}} \idty).
\end{equation*}
Now, using the expression for a function of a $2 \times 2$ matrix in terms of its trace and determinant 
\begin{equation*}
\sqrt X = \frac{(X + \sqrt{\det X} \idty)}{\tr{\sqrt X}},
\end{equation*}
one can define the ansatz state
\begin{equation*}
\omega_{123} := W = \sum_{i,j=1}^K |ii\>\<jj| \otimes \sqrt{p_i p_j} \sqrt{\rho_i \rho_j}\; \tr \sqrt{\rho_i \rho_j}.
\end{equation*}
Since $\omega_{123}$ is positive, so is its partial traces. The proof of part~\emph{1.} of the proposition follows from realizing that $\omega_{12} = \s C_{\s F}(\c E_K)$, and therefore positive. The second part follows from the application of the strong sub-additivity  to the tripartite state $\omega_{123}$.
\end{proof}

\section{Multi-state correlations \label{multi}}

Up to this point the correlation matrices used were taking into account only two states correlations. In an ensemble with more than two states it is natural to think on measures for \emph{multi-state correlations}. One can obtain bounds for the Holevo quantity in terms of multi-state correlations by considering special instances of correlation matrices~(\ref{cor}). The unitary matrices $U_i$ can be chosen in such a way that the entries of the correlation matrix on the sub- and super-diagonal are equal to the root fidelities $\sqrt{\s F_{ij}}$.

\begin{proposition}
\label{ppp2}
Consider an ensemble $\c E_K = \{(p_i, \rho_i)\}$ of $K$ faithful states of arbitrary dimension. The Holevo information $\chi( \c E_K)$ is bounded by the exchange entropy $\s S(\sigma)$,
\beq
\chi(\c E_K) \le \s S(\sigma),
\eeq
where the entries of the correlation matrix $\sigma$ are given by:
\beq 
\sigma_{ij} = \sqrt{p_ip_j}\, \tr \sqrt{\rho_j \rho_{j-1}}\, \frac{1}{\rho_{j-1}}\, \sqrt{\rho_{j-1} \rho_{j-2}}\, \frac{1}{\rho_{j-2}} \cdots \frac{1}{\rho_{i+1}}\, \sqrt{\rho_{i+1} \rho_{i}}.
\label{rec2}
\eeq
\end{proposition}

\begin{proof}
Consider the polar decomposition of the product of two matrices
\beq
\sqrt{\rho_i} \sqrt{\rho_j} = \abs{ \sqrt{\rho_i} \sqrt{\rho_j}} V_{i,j} = \sqrt{\rho_i^{1/2} \rho_j \rho_i^{1/2}} V_{i,j}.
\eeq
The Hermitian conjugate $V_{i,j}^\dagger$ of $V_{i,j}$ can be written as
\beq
V_{i,j}^\dagger = \frac{1}{\sqrt{\rho_j}}\, \frac{1}{\sqrt{\rho_i}}\,  \sqrt{\rho_i^{1/2} \rho_j \rho_i^{1/2}}.
\label{polarmixed}
\eeq
The unitaries $U_i$ in a correlation matrix~(\ref{cor}) are chosen according to
\beq
 U^\dagger_i = V_{i-1,i}^\dagger\, U^\dagger_{i-1}
\label{recurrence}
\eeq
where $V_{i-1,i}^\dagger$ is the unitary matrix from the polar decomposition and the first unitary $U_1$ can be chosen arbitrarily. The recurrence relation allows us to obtain formula~(\ref{rec2}).
\end{proof}

The matrix $\sigma$ has a layered structure, as shown below for $K=4$. 
\begin{equation}
\sigma=
\begin{bmatrix}
    p_1 & 0 & 0 &0 \\
    0 & p_2 &0 & 0 \\
    0 & 0 & p_3 &0\\
   0 &0&0 & p_4  
\end{bmatrix}+
\begin{bmatrix}
    0 &  f_{12} & 0 &0 \\
      f_{21} &0 &  f_{23} & 0 \\
    0 &  f_{32} &0 &  f_{34}\\
   0 &0& f_{43} & 0  
\end{bmatrix}+\begin{bmatrix}
    0 &0 & f^{(2)}_{13} & f^{(3)}_{14} \\
  0 &0 & 0 & f^{(2)}_{24} \\
    f^{(2)}_{31} &0 & 0 & \s0\\
    f^{(3)}_{41} & f^{(2)}_{42} &0 & 0  
\end{bmatrix}
\label{mx}
\end{equation}
On the main diagonal, the weights $\{p_j\}$ in the ensemble appear. Next, we find on the closest upper parallel to the main diagonal weighted root fidelities $\{f_{i,i+1}=\sqrt{p_i p_{i+1}} \tr \sqrt{\rho_i \rho_{i+1}}\}$. As we move on to more outward parallels $\{f^{(2)}_{i,i+2}\}$ and $\{f^{(3)}_{i,i+3}\}$ the matrix entries become more complicated and involve 3, 4, \dots  states. The matrix entries below the diagonal are the complex conjugates of these of the upper diagonal part. Note that the entropy of the matrix $\sigma$ depends on the ordering of the diagonal elements.  Proposition~\ref{ppp2} holds for any such ordering, and one can thus take the smallest entropy among all possible permutations. 

Using~(\ref{normeq}) it comes as no surprise that Proposition~\ref{ppp2} can be extended to non faithful states, properly defining the $\sigma_{ij}$. For the extreme case of pure states, $\rho_i = |\varphi_i\>\<\varphi_i|$ one puts
\beq
\sigma_{ij} = \sqrt{p_ip_j} \bra{\varphi_i} \varphi_j\>\, \r e^{-i\alpha_{i,j+1}}\, \r e^{-i\alpha_{i+1,i+2}} \cdots \r e^{-i\alpha_{j-1,j}},
\eeq
with 
\beq
\r e^{-i\alpha_{i,j}} := \frac{\bra{\varphi_i} \varphi_j \>}{\abs{\!\bra{\varphi_i} \varphi_j \>}}.
\eeq


\bigskip
\noindent
\textbf{Acknowledgements}

It is a pleasure to thank G. Mitchison and R. Jozsa for helpful discussions. This work was partially supported by the grant number N202 090239 of Polish Ministry of Science and Higher Education, by the Belgian Interuniversity Attraction Poles Programme P6/02, and by the FWO Vlaanderen project G040710N.

\bigskip
\noindent
\textbf{\Large Appendix}

\medskip
\begin{proof}[Proof of Lemma~\ref{lem2}]
As the supremum is always non-negative, we may impose the additional restriction $h \in \spa(\{f,g\})$. If $h$ does not belong to this subspace, decompose it into $h_1 \oplus h_2$ with $h_1 \in \spa(\{f,g\})$. Next replace $h$ by $h_1/\norm{h_1}$. Evaluating the functional with this new $h$ will return a value at least as large as that with the original $h$.

So let $h = \alpha f + \beta g$ with $\alpha, \beta \in \Cx$ such that
\beq
\norm h^2 = \abs\alpha^2 + \abs\beta^2 + 2 \Re\g e(\overline\alpha \beta \< f \,,\, g \>) = 1.
\label{nor}
\eeq
Using this normalization condition we compute
\begin{equation*}
\abs{\< f \,,\, h \>}^2 = \abs{\alpha + \beta \< f \,,\, g \>}^2 = 1 - \abs\beta^2 \bigl( 1 - \abs{\< f \,,\, g \>}^2 \bigr)
\end{equation*}
and
\begin{equation*}
\abs{\< g \,,\, h \>}^2 = \abs{\overline\alpha \< f \,,\, g \> + \overline\beta}^2 = 1 - \abs\alpha^2 \bigl( 1 - \abs{\< f \,,\, g \>}^2\bigr).
\end{equation*}
Hence, the functional of $h$ we have to maximize does not depend on the phase of $\overline\alpha \beta \< f \,,\, g \>$. The normalization condition~(\ref{nor}) can be satisfied if and only if
\begin{equation*}
\Bigl| \abs\alpha^2 + \abs\beta^2 -1 \Bigr| \le 2 \abs\alpha \abs\beta \abs{\< f \,,\, g \>}.
\end{equation*}
Putting $\lambda := \abs\alpha$, $\mu := \abs\beta$ and $t := \abs{\< f \,,\, g \>}$ we have to compute
\begin{equation*}
I := \sup_{\lambda, \mu} \Bigl( 2 - (\lambda^2 + \mu^2) (1 - t^2) - 2a \sqrt{1 - \lambda^2(1 - t^2)} \sqrt{1 - \mu^2 (1 - t^2)} \Bigr)
\end{equation*}
subject to the constraints
\begin{equation*}
0 \le \lambda,\enskip 0 \le \mu,\enskip \text{and } \abs{\lambda^2 + \mu^2 - 1} \le 2 \lambda \mu t
\end{equation*}
with $t$ satisfying $0 \le t \le a \le 1$. The supremum is attained choosing $\lambda = \mu$, with $\lambda$ such that
\begin{equation*}
\frac{1}{2(1 + t)} \le \lambda^2 \le \frac{1}{2(1 - t)}.
\end{equation*}
We obtain for $I$ the value
\begin{equation*}
I = 2 (1 - a) \Bigl(1 -\frac{1 - t^2}{2(1 + t)} \Bigr) = (1 - a) (1 + t) \le (1 - a^2).
\end{equation*}
\end{proof}

\end{document}